\documentclass[leqno,12pt]{article}

\usepackage{mathrsfs}  
\usepackage{amsmath}
\usepackage{amssymb}          
\usepackage{graphicx}               
\usepackage{latexsym}              
\usepackage{amsfonts}  
\usepackage{amsthm}

\newtheorem{theorem}{Theorem}      
\newtheorem{lemma}{Lemma}     
\newtheorem{definition}{Definition} 
\newtheorem{corollary}{Corollary}  

\newtheorem{remark}{Remark}
\DeclareMathOperator{\sech}{sech}

\newcommand{\f}{\frac}

\numberwithin{equation}{section}
\numberwithin{theorem}{section}
\numberwithin{lemma}{section}
\numberwithin{definition}{section}
\numberwithin{corollary}{section}

\begin{document}
\centerline{{\bf \large Sharp Upper Bound}}
\vspace{.1in}
\centerline{{\bf \large for Amplitudes of Finite-Gap Solutions}}
\vspace{.1in}
\centerline{{\bf \large of the Modified Korteweg-de Vries Equation}}

\vspace{.2cm}
\centerline{{Otis C. Wright, III}}
\centerline{{\em School of Science and Mathematics}}
\centerline{{\em Cedarville University}}
\centerline{{\em 251 N. Main St.}}
\centerline{{\em Cedarville, Ohio 45314}}
\centerline{E-mail: wrighto@cedarville.edu}
\centerline{{\today}}

\vspace{.2cm}

\centerline{{\bf \large Abstract}}
A direct proof based on commuting finite-dimensional flows and local polynomial invariants is given for a sharp upper bound on the amplitudes of finite-gap solutions of the modified Korteweg–de Vries (mKdV) equation. The maximal amplitude is  the sum of the imaginary parts of the upper-half-plane square roots of the roots of the invariant polynomial of the finite-gap solution of the focusing mKdV equation. An analogous formula is established for a bounded class of solutions of the defocusing mKdV equation. The upper bounds  are sharp and are explicitly attained by suitable initial data.


\section{Introduction}
\label{introduction}

The modified Korteweg-de Vries (mKdV) equation is
\begin{equation}
u_t + 6 \sigma u^2 u_x + u_{xxx} =  0,
\label{mkdv}
\end{equation}
where $u=u(x,t)$ is a real-valued wave profile and the parameter $\sigma = \pm 1$ dictates the physical nature of the wave dynamics. The focusing case ($\sigma=1$) supports localized soliton and breather solutions propagating under modulationally unstable conditions, while the defocusing case  ($\sigma=-1$) supports modulationally stable wave propagation.

In this paper, we show that for  finite-gap (\(N\)-phase) solutions of the focusing mKdV equation the maximizing configuration induces a natural factorization of the invariant polynomial \(\mathscr{R}(\lambda)\) on the double cover \(\lambda = E^2\). Let
$\mathcal E^+$ be the set of upper-half-plane square roots of the roots of the invariant polynomial $\mathscr{R}(\lambda).$ Then
\begin{equation}\label{sharp}
|u(x,t)|
\le
\sum_{E \in\mathcal E^+}\Im \left(E\right).
\end{equation}
The bound is sharp in the sense that there exists a solution that attains the upper bound. An analogous formula is established for a bounded class of solutions of the defocusing mKdV equation.

The mKdV equation is an essential mathematical model in nonlinear physics and fluid dynamics that governs the propagation of weakly dispersive, weakly nonlinear long waves when quadratic nonlinearities vanish or are sub-dominant~\cite{ablo 81, ablo 11}. Historically, the mKdV equation played a monumental role in the birth of modern soliton theory; its close relationship to the Korteweg-de Vries (KdV) equation via the celebrated Miura transformation~\cite{miur 68} directly inspired the discovery of the Inverse Scattering Transform (IST) by Gardner, Greene, Kruskal, and Miura~\cite{gard 67}. Beyond its deep theoretical value in integrability, the equation describes a variety of physical phenomena, including large-amplitude internal waves in stratified fluids~\cite{grim 97} and rogue waves~\cite{chen 18}. The finite-gap solutions of the mKdV equation represent quasiperiodic $N$-phase wavetrains whose complex spatial and temporal configurations are classically parametrized by Riemann theta functions associated with an invariant hyperelliptic Riemann surface~\cite{belo 94, gesz 03}. Under degenerate limits where the spectral band gaps close or coalesce, these general finite-gap solutions smoothly reduce to localized, physically striking structures such as solitons, breathers, and dispersive shock waves~\cite{chow 16, el 16}.

The analytical construction of finite-gap solutions to integrable nonlinear wave equations via commuting systems of ordinary differential equations was pioneered by Its and Kotlyarov~\cite{its 76, itsk 76, kotl 76}, and subsequently refined by Kamchatnov into an efficient, purely dynamical framework~\cite{kamc 00}. By utilizing the underlying Lax pair symmetries to bypass global transcendental machinery, this algebraic integration method effectively maps the wave dynamics onto local matrix polynomials. This framework was recently adapted to establish sharp upper bounds on the amplitudes of finite-gap solutions within three NLS-type hierarchies~\cite{wrig 19, wrig 20, wrig 24}.  A maximal amplitude formula for finite-gap solutions of the focusing NLS equation was also obtained by Bertola and Tovbis~\cite{bert 17} using algebro-geometric methods. In this paper, the finite-gap solutions of real mKdV equations are  studied on a loop algebra of $\mathfrak{sl}(2, \mathbb{C})$ with reality constraints and a shifted polynomial ansatz~\cite{gesz 03, gesz 00, bern 25}. 

The modified Korteweg–de Vries equation is closely connected to the sine-Gordon equation through the combined  sine-Gordon/mKdV hierarchy~\cite{gesz 03, gesz 00}. The finite-dimensional construction developed here extends naturally to the sine–Gordon equation through a different hierarchy reduction. The corresponding sharp density bounds will be investigated elsewhere.

A key conceptual advancement in the present work lies in a major simplification of the sharpness proof. In previous applications of the commuting ODE framework to  nonlinear Schrödinger hierarchies~\cite{wrig 19, wrig 20, wrig 24}, global algebro-geometric theta functions and compact tori  were invoked to guarantee that every individual continuous solution would attain its maximum. The current paper demonstrates that proving global compactness for every trajectory is not necessary for the determination of the sharp amplitude bound. Instead of requiring all solutions to achieve their peaks, we show that it is sufficient to establish the existence of a single solution that realizes the absolute supremum of all bounded solutions as its maximum. By showing that this supremum of the permissible initial data is explicitly attained by a constructible initial data point, we achieve a considerable simplification compared to the global topological arguments used previously.

The remainder of this paper is organized as follows. Section 2 establishes a concrete algebraic construction of finite-gap mKdV solutions using dynamical variables governed by autonomous, commuting systems of ordinary differential equations defined on a  loop algebra of $\mathfrak{sl}(2,\mathbb{C})$. Section 3 details the direct analysis at the level of the dynamical variables of the critical points of these commuting flows, culminating in the derivation of the sharp maximal amplitude formulas for the mKdV equation~(\ref{mkdv}). Section 4 presents degenerate one-gap limits that connect the sharp amplitude formulas with the classical soliton and kink solutions.

\section{Finite-Gap Solutions}

In this work the real-valued mKdV equation~\eqref{mkdv} is studied within a loop algebra of $\mathfrak{sl} (2,\mathbb{C})$ with reality constraints and a shifted polynomial ansatz. Under this ansatz, the wave profile $u(x,t)$ is associated directly with the diagonal entry of a generating matrix Laurent polynomial $\Psi^{(N)}.$ The Laurent polynomials used to define the commuting hierarchy flows satisfy a shifted ansatz with respect to the spectral parameter. A nonhomogeneous splitting of the loop algebra is used to ensure that the flows are well-defined and commute. 

The finite-gap solutions are constructed from solutions of $N$ nonlinear autonomous ordinary differential equations generated by $3N$ dependent complex dynamical variables. The invariant reality conditions reduce the phase space to $3N$ real parameters~\cite{bern 25}. These $N$ ordinary differential equations commute with each other and possess a common set of invariants, allowing the compatible local solution to extend to a compatible global solution. In earlier treatments of the NLS, derivative NLS, and modified NLS hierarchies, the sharp upper amplitude bounds were established by maximizing the compatible global solution over a compact $N$-dimensional torus where continuous trajectories were guaranteed to achieve a maximum~\cite{wrig 19, wrig 20, wrig 24}. By contrast, the current paper avoids the need to show that every bounded trajectory attains its maximum, defining the sharp upper bound directly as the supremum over the class of real bounded finite-gap initial data. This supremum is explicitly achieved by the amplitude of a solution of the mKdV equation for specific initial conditions.

\begin{definition}[Loop Algebra and Splitting Operators]
Let $\mathfrak{g} = \mathfrak{sl}(2,\mathbb{C}) \otimes \mathbb{C} [\lambda, \lambda^{-1}],$ the loop algebra of Laurent polynomials in the spectral parameter $\lambda,$ and let $\Phi \in \mathfrak{g},$ viz.,
\begin{equation}
\Phi = \sum\limits^{M_{+}}_{j=M_{-}} A_j \lambda^j,
\end{equation}
where $A_j \in \mathfrak{sl}(2,\mathbb{C}),$ for $j \in \mathbb{Z}.$ 
Write 
\begin{equation}
A_j = \left(\begin{array}{cc} - a_j & b_j \\ c_j & a_j \end{array}\right),
\end{equation}
for $j=M_{-}, \ldots, M_{+}.$
For any $\Phi \in \mathfrak{g},$ define the splitting operators
\begin{equation}
(\Phi)_+ = \sum\limits^{M_+}_{j=0} A_j \lambda^j - \left(\begin{array}{cc} 0 & 0 \\ c_0 & 0 \end{array} \right)
\end{equation}
and
\begin{equation}
(\Phi)_- = \sum\limits^{-1}_{j=M_{-}} A_j \lambda^j  + \left(\begin{array}{cc} 0 & 0 \\ c_0 & 0 \end{array} \right).
\end{equation}

\end{definition}

The splitting operators project $\mathfrak{g}$ onto two Lie subalgebras $\mathfrak{g}_+$ and $\mathfrak{g}_-,$ creating a splitting of the loop algebra
\begin{equation}
\mathfrak{g} = \mathfrak{g}_+ \oplus \mathfrak{g}_{-}
\end{equation}
that enables the construction of the commuting hierarchy flows.
The correction term involving \(c_0\) preserves the closure of both \(\mathfrak g_+\) and \(\mathfrak g_-\) under the commutator bracket, and  it also ensures that the commuting hierarchy flows defined later, using a shifted polynomial ansatz, close on the set of dynamical variables.  This construction is closely related to the Adler-Kostant-Symes framework for integrable hierarchies~\cite{adle 78, kost 79, syme 80}
and its loop-algebra generalizations~\cite{drin 85, kac 90}. An explicit construction of the finite-gap solutions based on this framework appears in the work of Bernatska~\cite{bern 25}. 

The splitting operators commute with an involution of the Lie algebra that we use later to express the real focusing mKdV equation reduction of the phase space.

\begin{lemma}[Commutation of the Splitting and Involution] \label{involutionlemma}
Define the involution $\tau: \mathfrak g \rightarrow \mathfrak g$ by
\begin{equation}
\left( \tau \Phi \right) (\lambda) = J(\lambda) \Phi (\lambda)^T J^{-1} (\lambda),
\end{equation}
where $(\cdot)^T$ is the matrix transpose and 
\begin{equation}
J(\lambda) = \left(\begin{array}{cc} 1 & 0 \\ 0 & \lambda \end{array} \right).
\end{equation}
Then the involution $\tau$ commutes with the projection operators, viz., 
\begin{equation}
\tau \left( \Phi \right)_{\pm}  = \left(   \tau \Phi  \right)_{\pm}. 
\end{equation}
\end{lemma}

\begin{proof}
For any $\Phi \in \mathfrak{g},$ let $(\cdot)_{0}$ be the usual homogeneous splitting operator,
\begin{equation}
\left( \Phi \right)_{0} = \sum\limits^{M_+}_{j=0} A_j \lambda^j.
\end{equation}
The splitting operator $( \cdot )_{0}$ commutes with the transpose and, also, we can write
\begin{equation}
\left( \Phi \right)_{+} = \left( \Phi \right)_{0} - \left( \begin{array}{cc} 0 & 0 \\ c_0 & 0 \end{array} \right).
\end{equation}
Then
\begin{equation}
\begin{array}{rcl}
 \tau \left( \Phi \right)_{+}  & = &   \tau  \left( \Phi \right)_0 - \tau \left( \begin{array}{cc} 0 & 0 \\ c_0 & 0 \end{array} \right) \\
& = & J \left( \Phi^T \right)_0 J^{-1} - \left(\begin{array}{cc} 0 & c_0 \lambda^{-1} \\ 0 & 0 \end{array} \right) \\
& = & \left( J \Phi^T J^{-1} \right)_{0} + \left(\begin{array}{cc} 0 & c_0 \lambda^{-1} \\ - b_{-1} & 0 \end{array} \right) - \left(\begin{array}{cc} 0 & c_0 \lambda^{-1} \\ 0 & 0 \end{array} \right) \\
& =&  \left( J \Phi^T J^{-1} \right)_{0} - \left(\begin{array}{cc} 0 & 0  \\  b_{-1} & 0 \end{array} \right) \\
& = & \left( J \Phi^T J^{-1} \right)_{+} \\
& = & \left( \tau \Phi  \right)_{+},
\end{array}
\end{equation}
since $b_{-1}$ becomes the lower-left coefficient of degree zero in the matrix $J \Phi^T J^{-1}.$
Consequently, 
\begin{equation}
\tau \left(\Phi \right)_{-}  = \tau  \Phi - \tau \left( \Phi \right)_+ = \tau \Phi  - \left( \tau \Phi \right)_+ =  \left( \tau \Phi  \right)_{-}.
\end{equation}
\end{proof}

\begin{corollary}[Commutation of the Splitting and Derived Involution]
Define the involution $\rho: \mathfrak g \rightarrow \mathfrak g$ by
\begin{equation}
(\rho\Phi)(\lambda) =-\overline{(\tau\Phi)(\overline{\lambda})}.
\end{equation}
Then the involution $\rho$ commutes with the projection operators, viz.,
\begin{equation}
\rho (\Phi)_+ = \bigl(\rho \Phi \bigr)_+,
\qquad
\rho (\Phi)_-
=
\bigl(\rho \Phi \bigr)_-.
\end{equation}
\end{corollary}

\begin{proof}
By Lemma~\ref{involutionlemma}
\begin{equation}
\tau (\Phi)_\pm
=
\bigl(\tau \Phi \bigr)_\pm .
\end{equation}
Since complex conjugation commutes with the projection operators,
\begin{equation}
-\overline{
\left(\tau \left( \Phi \right)_{\pm} \right) (\overline{\lambda}) 
}
=
-\overline{
\bigl( \left( \tau \Phi \right)(\overline{\lambda} )\bigr)_\pm
}
=
\left(
-\overline{\left( \tau \Phi \right) (\overline{\lambda})}
\right)_\pm .
\end{equation}
Hence
\begin{equation}
\rho (\Phi)_\pm 
=
\bigl(\rho \Phi \bigr)_\pm .
\end{equation}
\end{proof}

\begin{definition}[Dynamical Variables] For $N \in \mathbb{N},$ consider $3N$ complex-valued variables $f_j,g_j, h_j \in \mathbb{C},$ for $j=0, \ldots, N-1.$ Define three polynomials
\begin{equation}
F_{N-1}(\lambda) =  \sum\limits_{j=0}^{N-1} f_j \lambda^j
\end{equation}
\begin{equation}
G_{N}(\lambda) = \sum\limits_{j=0}^{N} g_j \lambda^j,
\end{equation}
\begin{equation}
H_{N} (\lambda) = \sum\limits_{j=0}^{N} h_j \lambda^j,
\end{equation}
where $g_{N} = h_{N} = b \in \mathbb{C}$ is constant.
Then define the generating matrix
\begin{equation}
\Psi^{(N)} (\lambda) = \left(\begin{array}{cc} -F_{N-1}(\lambda) & G_{N}(\lambda) \lambda^{-1} \\ H_{N} (\lambda) & F_{N-1}(\lambda) \end{array} \right),
\end{equation}
where  $\Psi^{(N)} (\lambda) \in \mathfrak{g} = \mathfrak{sl}(2,\mathbb{C}) \otimes \mathbb{C}[\lambda, \lambda^{-1}].$
\end{definition}

The generating matrix respects the parity of the diagonal and off-diagonal entries of an element of $\mathfrak{g}$ but shifts the degrees of the off-diagonal elements relative to each other. The relative shift of the off-diagonal polynomials is crucial to the construction of  the real mKdV hierarchy and is the reason for the correction term in the projection operators.

\begin{definition}[Time Flow Matrices]
Define the hierarchy flow matrices
\begin{equation}
V^{(k)}  = ( \lambda^{k-N+1} \Psi^{(N)} )_+,
\end{equation}
for $k=0, \ldots, N-1.$  
\end{definition}

Now we define the flows on the dynamical variables.

\begin{definition}[Nonlinear Autonomous Systems]
Define $N$ nonlinear autonomous ordinary differential equations with $N$ independent time variables $t_k \in \mathbb{R},$ for $k=0, \ldots, N-1,$ and $3N$ dependent complex variables
given by the components of $\mathbf{y} \in \mathbb{C}^{3N},$ viz., $y_j=f_{j-1}, y_{N+j}= g_{j-1}, y_{2N+j}= h_{j-1},$ for $j=1, \ldots, N,$  by
\begin{equation}
\frac{\partial}{\partial t_k} \mathbf{y} = \mathbf{F} (\mathbf{y}),
\label{Fode}
\end{equation}
where  
the components of $\mathbf{F} : \mathbb{C}^{3N} \rightarrow \mathbb{C}^{3N}$ are  polynomials  in the $3N$ dependent variables, $\mathbf{y} \in \mathbb{C}^{3N},$ determined by equating the coefficients of powers of $\lambda$ in the equations
\begin{equation}
\f{\partial \Psi^{(N)}}{\partial t_k} = \left[V^{(k)},\Psi^{(N)} \right],
\label{ode}
\end{equation}
where 
\begin{equation}
[A,B] = AB - BA
\end{equation}
is the usual commutator operator on arbitrary matrices $A, B \in \mathfrak{sl}(2,\mathbb{C}).$
\end{definition}

The hierarchy equations in~(\ref{ode}) define a closed finite-dimensional dynamical system on the Laurent-polynomial ansatz \(\Psi^{(N)}\). The relative shift in the off-diagonal elements and the projection operator act together to close the recursion in powers of $\lambda.$ 
In particular, the generating matrix $\Psi^{(N)}$ has the Laurent expansion
\begin{equation}\label{psiexpansion}
\Psi^{(N)}= \left(\begin{array}{cc} 0 & 0 \\ b & 0 \end{array} \right) \lambda^N + \left(\begin{array}{cc} -f_{N-1} & b \\ h_{N-1} & f_{N-1} \end{array} \right) \lambda^{N-1} + \cdots + \left( \begin{array}{cc} -f_0 & g_1 \\ h_0 & f_0 \end{array}\right) + \left(\begin{array}{cc} 0 & g_0 \\ 0 & 0 \end{array}\right) \lambda^{-1},
\end{equation}
and 
\begin{equation}\label{Vpsiexpansion}
\begin{array}{rcl}
\left[ V^{(k)}, \Psi^{(N)} \right] &= & \left[ \left(\begin{array}{cc} 0  & 0 \\ b & 0 \end{array}\right) \lambda^{k+1} + \left(\begin{array}{cc} - f_{N-1} & b \\ h_{N-1} & f_{N-1} \end{array} \right) \lambda^k + \ldots+ \left(\begin{array}{cc} -f_{N-k-1} & g_{N-k} \\ 0 & f_{N-k-1} \end{array} \right),\right. \\[.2in]
& & \left. \left(\begin{array}{cc} 0 & 0 \\ h_{N-k-1} & 0 \end{array}\right) \lambda^{N-k-1} + \left(\begin{array}{cc} -f_{N-k-2} & g_{N-k-1} \\ h_{N-k-2} & f_{N-k-2} \end{array}\right) \lambda^{N-k-2} + \ldots \right. \\[.2in]
&& \left. + \left( \begin{array}{cc} 0 & g_0 \\ 0 & 0 \end{array} \right) \lambda^{-1} \right] \\[.2in]
& = &\left(\begin{array}{cc}-b g_{N-k-1}+b h_{N-k-1} &  0 \\ -2 b f_{N-k-2} +2 f_{N-1} h_{N-k-1} & b g_{N-k-1} - b h_{N-k-1} \end{array}\right) \lambda^{N-1} + \\[.2in]
& & M_{N-2} \lambda^{N-2}+ \cdots + M_0 + \left(\begin{array}{cc} 0 & -2 f_{N-k-1} g_0 \\ 0 & 0 \end{array}\right) \lambda^{-1},
\end{array}
\end{equation}
where $M_{0}, \ldots, M_{N-2} \in \mathfrak{sl} (2,\mathbb{C})$ have linear or quadratic entries in the dynamical variables. By equating coefficients of equal powers of \(\lambda\) and comparing matrix entries, we obtain a closed autonomous system on the dynamical variables.
The projection operator $(\cdot)_+$ forces the lowest order matrix in the expansion of $V^{(k)}$ to be upper triangular. This projection mimics the relative shift of the loop algebra and is crucial to closing the recursion relations on the dynamical variables.

The splitting of the loop algebra $\mathfrak{g} = \mathfrak{g}_+ \oplus \mathfrak{g}_{-}$ and the definition of the time flow matrices as projections of the generating matrix $\Psi^{(N)}$ produce the following zero-curvature identity.

\begin{lemma}[Zero-Curvature Identity~\cite{dick 03}]
The $N$ matrix polynomials $V^{(k)},$ for $k = 0, \ldots, N-1,$ satisfy the identity
\begin{equation}
\frac{\partial }{\partial t_\ell} V^{(k)} - \frac{\partial }{\partial t_k} V^{(\ell)}= [V^{(\ell)}, V^{(k)}],
\label{Videntity}
\end{equation}
for $k, \ell = 0, \ldots, N-1.$
\end{lemma}
\begin{proof}
\begin{equation}
\begin{array}{rcl}
\frac{\partial}{\partial t_\ell} V^{(k)} - \frac{\partial}{\partial t_k}  V^{(\ell)} & = & \frac{\partial}{\partial t_\ell} \left(\lambda^{k-N+1} \Psi^{(N)} \right)_{+}  - \frac{\partial}{\partial t_k} \left(\lambda^{\ell - N+1} \Psi^{(N)} \right)_+ \\[.1in]
& = &\left(\lambda^{k-N+1} \frac{\partial}{\partial t_\ell} \Psi^{(N)} \right)_+ - \left(\lambda^{\ell -N+1} \frac{\partial}{\partial t_k} \Psi^{(N)} \right)_+ \\[.1in]
& = &\left(\lambda^{k-N+1} [V^{(\ell)}, \Psi^{(N)}]\right)_+ - \left(\lambda^{\ell - N+1} [ V^{(k)}, \Psi^{(N)}]\right)_+ \\[.1in]
& = & \left([V^{(\ell)}, \lambda^{k-N+1} \Psi^{(N)}] \right)_+ - \left([V^{(k)}, \lambda^{\ell - N+1} \Psi^{(N)}] \right)_+ \\[.1in]
& = & \left([(\lambda^{\ell - N+1} \Psi^{(N)})_+, \lambda^{k-N+1} \Psi^{(N)}]\right)_+ - \\[.1in]
&&  \left([(\lambda^{k-N+1} \Psi^{(N)})_+, \lambda^{\ell - N+1} \Psi^{(N)}]\right)_+ \\[.1in]
& = & \left([\lambda^{\ell - N+1} \Psi^{(N)} - (\lambda^{\ell - N+1} \Psi^{(N)})_-, \lambda^{k-N+1} \Psi^{(N)}] \right)_+ - \\[.1in]
&&   \left([(\lambda^{k-N+1} \Psi^{(N)})_+, \lambda^{\ell - N+1} \Psi^{(N)}]\right)_+ \\[.1in]
& = & \left([\lambda^{k-N+1} \Psi^{(N)}, (\lambda^{\ell - N+1} \Psi^{(N)})_-]\right)_+ - \\[.1in]
& & \left([(\lambda^{k-N+1} \Psi^{(N)})_+, \lambda^{\ell - N+1} \Psi^{(N)}] \right)_+ \\[.1in]
& = & \left([(\lambda^{k-N+1} \Psi^{(N)})_+, (\lambda^{\ell - N+1} \Psi^{(N)})_-]\right)_+ - \\[.1in]
&& \left([(\lambda^{k-N+1} \Psi^{(N)})_+, \lambda^{\ell - N+1} \Psi^{(N)}]\right)_+ \\[.1in]
& = & \left([ (\lambda^{k-N+1} \Psi^{(N)})_+, (\lambda^{\ell - N+1} \Psi^{(N)})_- - \lambda^{\ell - N+1} \Psi^{(N)}]\right)_+ \\[.1in]
& = & \left[(\lambda^{\ell - N+1} \Psi^{(N)})_+, (\lambda^{k-N+1} \Psi^{(N)})_+ \right] \\[.1in]
& = & \left[V^{(\ell)}, V^{(k)}\right].
\end{array}
\end{equation}
\end{proof}

Using the zero-curvature  identity~(\ref{Videntity}), the  definition of the time flows guarantees that the time flows commute in a classical sense. For ease of notation, define 
$
\partial_k = \frac{\partial}{\partial t_k}$ and $\partial_\ell = \frac{\partial}{\partial t_\ell}.$

\begin{lemma}[Commuting Differential Equations~\cite{dick 03}]
The $N$ nonlinear autonomous ordinary differential equations defined by equation~(\ref{ode}) commute with each other, viz.,
\begin{equation}
\partial_k (\partial_\ell \Psi^{(N)}) = \partial_\ell (\partial_k \Psi^{(N)}),
\end{equation}
for $k, \ell = 0, 1, \ldots, N-1.$
\end{lemma}
\begin{proof}
Using the Jacobi identity and the zero-curvature identity~(\ref{Videntity}), 
\begin{equation}
\begin{array}{rcl}
\partial_k (\partial_\ell \Psi^{(N)}) & = & \partial_k [V^{(\ell)},\Psi^{(N)}] \\[.1in]
& = &[\partial_k V^{(\ell)}, \Psi^{(N)}] + [V^{(\ell)}, \partial_k \Psi^{(N)}]  \\[.1in]
& = & [\partial_k V^{(\ell)}, \Psi^{(N)}] + [V^{(\ell)}, [V^{(k)}, \Psi^{(N)}]]  \\[.1in]
& = & [\partial_k V^{(\ell)}, \Psi^{(N)}]  - [ \Psi^{(N)}, [V^{(\ell)}, V^{(k)}]] - [V^{(k)}, [ \Psi^{(N)}, V^{(\ell)}]] \\[.1in]
& = & [\partial_k V^{(\ell)} + [ V^{(\ell)}, V^{(k)}], \Psi^{(N)}]  + [V^{(k)}, [V^{(\ell)}, \Psi^{(N)}]] \\[.1in]
& = & [\partial_\ell V^{(k)}, \Psi^{(N)}]+[V^{(k)}, \partial_\ell \Psi^{(N)}] \\[.1in]
& = & \partial_\ell [ V^{(k)}, \Psi^{(N)}] \\[.1in]
& = & \partial_\ell (\partial_k \Psi^{(N)}).
\end{array}
\end{equation}
\end{proof}

Thus the hierarchy generated by $\Psi^{(N)}$ defines compatible finite-dimensional dynamical systems of classical commuting flows. In order to produce real solutions of the 
mKdV equation~(\ref{mkdv}), we consider invariant reality conditions on the dynamical variables. The focusing reality conditions are associated with the real form $\mathfrak{su}(2),$ and the defocusing reality conditions are associated with  the real form $\mathfrak{sl}(2,\mathbb{R}).$

\begin{theorem}[Reality Conditions]
Suppose the initial data satisfy either one of the following reality conditions:
\begin{enumerate}
\item Focusing reality conditions,
\begin{equation}\label{focusingreality}
f_j=-\overline{f_j},
\qquad
g_j=-\overline{h_j},
\qquad
b = i,
\end{equation} for $j=0, \ldots, N-1.$
\item Defocusing reality conditions,
\begin{equation}\label{defocusingreality}
f_j = \overline{f_j}, 
\qquad
g_j =\overline{g_j},
\qquad
h_j=\overline{h_j},
\qquad
b = 1,
\end{equation} for $j=0, \ldots, N-1.$
\end{enumerate}
Then the hierarchy flows preserve the reality conditions.
\end{theorem}

\begin{proof}

If we set $b=i,$ then the focusing reality conditions~(\ref{focusingreality}) are equivalent to 
\begin{equation}
\left( \rho  \Psi^{(N)} \right) (\lambda)  = - J(\lambda)  \left(\Psi^{(N)} (\overline{\lambda}) \right)^\dagger J^{-1}(\lambda) = \Psi^{(N)} (\lambda).
\end{equation}
where $(\cdot)^\dagger$ indicates the conjugate transpose. Now, for $k=0, 1, \ldots, N-1,$
\begin{equation}
\begin{array}{rcl}
\frac{\partial}{\partial t_k} \left(  \rho \Psi^{(N)}  \right)     
& = & - J  \left(\left[V^{(k)} (\overline{\lambda}),\Psi^{(N)} (\overline{\lambda}) \right]\right)^\dagger J^{-1}  \\[.1in]
& = & \left[ -J \left( V^{(k)}(\overline{\lambda}) \right)^\dagger J^{-1}, - J \left( \Psi^{(N)} (\overline{\lambda}) \right)^\dagger J^{-1} \right] \\[.1in]
& = & \left[  \rho V^{(k)} ,  \rho \Psi^{(N)}  \right] \\[.1in]
& = &  \left[  \rho  \left(\lambda^{k-N+1} \Psi^{(N)}\right)_{+}  ,  \rho \Psi^{(N)}  \right] \\[.2in]
& = &  \left[  \left(  \lambda^{k-N+1}  \left(\rho    \Psi^{(N)} \right) \right)_{+}  ,   \rho \Psi^{(N)}  \right]. \\[.2in]
\end{array}
\end{equation}
Since the transformed solution satisfies the same hierarchy equations and agrees with the original solution at the initial time whenever the focusing reality conditions hold, uniqueness implies that the focusing reality conditions are preserved for all hierarchy times.

If we set $b=1,$ then the defocusing reality conditions~(\ref{defocusingreality}) are equivalent to
\begin{equation}
\overline{ \Psi^{(N)} (\overline{\lambda} ) } = \Psi^{(N)}(\lambda).
\end{equation}
Now
\begin{equation}
\begin{array}{rcl}
\frac{\partial}{\partial t_k}  \left(  \overline{\Psi^{(N)} (\overline{\lambda})} \right)
& = & \left[\overline{V^{(k)} (\overline{\lambda})},\overline{\Psi^{(N)} (\overline{\lambda})} \right] \\
& = & \left[ \left( \lambda^{k-N+1} \overline{\Psi^{(N)} (\overline{\lambda} ) } \right)_{+}, \overline{\Psi^{(N)} (\overline{\lambda} )} \right].
\end{array}
\end{equation}
As before, the transformed solution satisfies the same hierarchy equations. If the solution agrees with the original solution at the initial time when the defocusing reality conditions hold, then uniqueness implies that the defocusing reality conditions are preserved for all hierarchy times.

\end{proof}

\begin{theorem}[mKdV Equation and Hierarchy Relations]
Let
$
\Psi^{(N)}
$
satisfy the hierarchy equations
\begin{equation} \label{hier}
\partial_k\Psi^{(N)}
=
[V^{(k)},\Psi^{(N)}],
\qquad
k=0,\ldots,N-1,
\end{equation}
subject to one of the two  invariant reality conditions~(\ref{focusingreality}) or~(\ref{defocusingreality}).

Then the leading-order recursion relations are
\begin{subequations}\label{leadingorder}
\begin{align}
\partial_k f_{N-1}
&=
b(g_{N-k-1}-h_{N-k-1}), \label{leadingA}\\
\partial_k g_{N-1}
&=
2 b f_{N-k-2}-2 f_{N-1} g_{N-k-1}, \label{leadingB}\\
\partial_k h_{N-1}
&=
-2 b f_{N-k-2} + 2 f_{N-1} h_{N-k-1},\label{leadingC}
\end{align}
\end{subequations}
for \(k=0,\ldots,N-1\), with the convention that $f_j=g_j=h_j=0,$ for $j <0.$

Moreover, if we set $f_{N-1}=iu,$ with focusing reality conditions, or $f_{N-1}=u,$ with defocusing reality conditions,  then the case $N=1$ reduces to the stationary mKdV equation for the $k=0$ flow, and the case $N \geq 2$ reduces to the mKdV equation~(\ref{mkdv}) for the $k=0$ and $k=1$ flows.
\end{theorem}

\begin{proof}
The leading-order recursion relations follow by equating coefficients in the hierarchy equation~(\ref{hier}) (compare equations~(\ref{psiexpansion}) and~(\ref{Vpsiexpansion})).
Equations~(\ref{leadingA}, \ref{leadingC}) arise from the diagonal and lower-left entries at the level of \(\lambda^{N-1}\), while equation~(\ref{leadingB}) arises from the upper-right entry at the level of \(\lambda^{N-2}\).

In the case where $N=1,$ there is only one time flow, $k=0,$ which we identify with the spatial variable $t_0=x.$ Using the fact that $f_{j}=g_{j}=h_{j}=0$ for all $j<0,$  the system of equations is simply
\begin{subequations}\label{genus0xeqns}
\begin{align}
f_{0x} &=  b (g_0-h_0) \label{genus0xeqnsA}\\[.1in]
g_{0x} &= - 2 f_0 g_0 \\[.1in]
h_{0x} &= 2 f_0 h_0 
\end{align}
\end{subequations}
Summing the derivatives of $g_0$ and $h_0$ and integrating gives 
\begin{equation}
b(g_0+h_0) = -f_0^2 + c,
\end{equation}
where $c \in \mathbb{R}$ is a real constant under both forms of the reality conditions. Differentiating the expression for $f_{0x}$ and eliminating $g_0$ and $h_0$ yields 
\begin{equation}
f_{0xx} - 2 f_0^3 + 2 c f_0 = 0,
\end{equation}
which can be differentiated to obtain
\begin{equation}\label{statgenus1}
2c f_{0x} - 6 f_0^2 f_{0x} + f_{0xxx} = 0.
\end{equation}
Under  the focusing reality conditions, define $u(X,T) = -i f_0 (X+2cT),$ where $x=X+2cT$ and $u$ is a real function, then equation~(\ref{statgenus1}) becomes the focusing mKdV equation~(\ref{mkdv}) with $\sigma=1,$
\begin{equation}
u_T+6 u^2 u_{X} + u_{XXX} = 0.
\end{equation}
Under the defocusing reality conditions, define $u(X,T) = f_0(X+2 cT)$
where $x=X+2cT$ and $u$ is a real function, then equation~(\ref{statgenus1}) becomes the defocusing mKdV equation~(\ref{mkdv}) with $\sigma= -1,$
\begin{equation}
u_T - 6 u^2 u_{X} + u_{XXX} = 0.
\end{equation}

If $N \geq 2,$ then the $k=0$ and $k=1$ flows together produce the mKdV equation with the identification $t_0=x$ and $t_1=t.$ The leading $k=0$ recursion relations  are
\begin{subequations} \label{zeroflow}
\begin{align}
\left( f_{N-1} \right)_x &=  b \left( g_{N-1} - h_{N-1} \right), \label{zeroflowa} \\[.1in]
\left(g_{N-1} \right)_x &=  2 b f_{N-2} - 2 f_{N-1} g_{N-1}, \label{zeroflowb} \\[.1in]
\left(h_{N-1} \right)_x &=   -2 b f_{N-2} + 2 f_{N-1} h_{N-1}. \label{zeroflowc}
\end{align}
\end{subequations}
The next level of the recursion of the $k=0$ flow yields
\begin{equation} \label{zeroflownext}
\left( f_{N-2} \right)_x =  b \left( g_{N-2} - h_{N-2} \right),
\end{equation}
while the leading equation  of the $k=1$ flow is
\begin{equation} \label{firstflow}
\left(f_{N-1} \right)_t = b \left( g_{N-2} - h_{N-2}\right).
\end{equation}
Consequently,
\begin{equation} \label{firsttflow}
\left(f_{N-1}\right)_t = \left( f_{N-2} \right)_x.
\end{equation}
To eliminate $f_{N-2},$ add equations~(\ref{zeroflowb}) and~(\ref{zeroflowc}). Using equation~(\ref{zeroflowa}) gives
\begin{equation}
\left( g_{N-1} + h_{N-1} \right)_x = -2 f_{N-1} \left(g_{N-1} - h_{N-1} \right)  = -\frac{1}{b} \left( f_{N-1}^2 \right)_x.
\end{equation}
Integrating yields
\begin{equation} \label{ghsum}
b(g_{N-1} + h_{N-1}) = - f_{N-1}^2 + c,
\end{equation}
where $c \in \mathbb{R}$ is a real constant under both forms of the reality conditions.

Next subtract equations~(\ref{zeroflowb}) and~(\ref{zeroflowc}). Using equation~(\ref{ghsum}) gives
\begin{equation}
b\left(g_{N-1} - h_{N-1} \right)_x = 4 b^2 f_{N-2} +2 f_{N-1}^3 - 2c f_{N-1}.
\end{equation}
Applying equation~(\ref{zeroflowa}) then yields
\begin{equation}\label{step1}
4 b^2 f_{N-2} = \left( f_{N-1} \right)_{xx} -2 f_{N-1}^3 + 2 c f_{N-1}.
\end{equation}
Differentiating equation~(\ref{step1}) and substituting into equation~(\ref{firsttflow}) produces 
\begin{equation}\label{eveqn}
4 b^2 \left( f_{N-1} \right)_t =  \left(f_{N-1}\right)_{xxx} - 6 f_{N-1}^2 \left( f_{N-1} \right)_x + 2c \left( f_{N-1} \right)_x.
\end{equation}

Under  the focusing reality conditions $b=i.$ Define the real-valued function
\begin{equation}\label{focsoln}
u(X,T) = -i f_{N-1} \left(X+2cT, 4 T \right),
\end{equation}
where $x=X+2cT$ and $t=4 T$. Then equation~(\ref{eveqn}) becomes
\begin{equation}
u_T+6u^2u_X+u_{XXX}=0,
\end{equation}
which is the focusing mKdV equation~(\ref{mkdv}) with $\sigma=1.$

Under  the defocusing reality conditions $b=1.$ Define the real-valued function
\begin{equation}\label{defsoln}
u(X,T) = f_{N-1}\left(X+2cT, -4  T \right), 
\end{equation}
where $x=X+2cT$ and $t=-4T$. Then equation~(\ref{eveqn}) becomes
\begin{equation}
u_T -6 u^2 u_X + u_{XXX} = 0,
\end{equation}
which is the defocusing mKdV equation~(\ref{mkdv}) with $\sigma=-1.$

\end{proof}

The shifted off-diagonal structure of the generating matrix $\Psi^{(N)}$ is reflected in the  recursion relations~(\ref{leadingorder}) and is the mechanism that distinguishes the mKdV hierarchy from the standard AKNS hierarchy~\cite{akns 73}.

\begin{lemma}[Invariant Polynomial]
The $2N+1$-degree polynomial $\mathscr{R}(\lambda),$ defined by
\begin{equation}
\mathscr{R}(\lambda) = -\lambda^2 \det \Psi^{(N)} (\lambda) = \lambda^2 F_{N-1}^2(\lambda) + \lambda G_{N}(\lambda) H_{N}(\lambda) = b^2 \lambda \sum_{j=1}^{2N} (\lambda-\lambda_{j}),
\label{Reqn}
\end{equation}
is  invariant with respect to all the time-flow variables $t_k,$ for $k=0, \ldots, N-1.$  The numbers $\lambda_j \in \mathbb{C},$ for $j=1, \ldots, 2N,$ and the fixed root at 
$\lambda =0$ are the
invariant roots of $\mathscr{R}(\lambda) = 0$ and are assumed to be distinct.
\label{invariant}
\end{lemma}
\begin{proof}
The invariance of the determinant of $\Psi$ follows from Jacobi's formula for the derivative of the determinant of a matrix.  If $\mbox{adj}  (\Psi)$ denotes the adjugate matrix of $\Psi,$ then using equation~(\ref{ode}), for $k=0, \ldots, N-1,$ and the invariance of the trace under cyclic permutations,
\begin{equation}
\begin{array}{rcl}
\partial_k \det (\Psi^{(N)}) & = & \mbox{tr}  \left( \mbox{adj} (\Psi^{(N)}) \partial_k \Psi^{(N)} \right) \\[.1in]
& = & \mbox{tr} \left(\mbox{adj} (\Psi^{(N)} ) (V^{(k)} \Psi^{(N)}  - \Psi^{(N)}  V^{(k) }\right)\\[.1in] 
& = & \mbox{tr} \left( \mbox{adj} (\Psi^{(N)} )  V^{(k)}  \Psi^{(N)}  - \mbox{adj} (\Psi^{(N)} )  \Psi^{(N)}   V^{(k)}\right) \\[.1in]
& = & \mbox{tr} \left( \mbox{adj} (\Psi^{(N)} ) V^{(k)} \Psi^{(N)}  \right) - \mbox{tr} \left(\mbox{adj}(\Psi^{(N)} )  \Psi^{(N)}   V^{(k)} \right) \\[.1in] 
& = & \mbox{tr} \left( \Psi^{(N)}   \mbox{adj} (\Psi^{(N)} ) V^{(k)} \right) - \mbox{tr} \left( \mbox{adj} (\Psi^{(N)} )  \Psi^{(N)}   V^{(k)}\right)\\[.1in]
& = & \mbox{tr} \left( \det ( \Psi^{(N)} ) V^{(k)} \right) - \mbox{tr} \left( \det ( \Psi^{(N)} ) V^{(k)} \right) \\
&= & 0.
\end{array}
\end{equation}
\end{proof}

Notice that the  invariant hyperelliptic spectral curve $\Gamma : w^2 = -\det \Psi^{(N)} = F_{N-1}(\lambda)^2 + G_{N} H_{N} \lambda^{-1}$ has a fixed branch point at
 $\lambda =0,$ in addition to the $2N$ roots of $\lambda F_{N-1} (\lambda)^2 + G_{N} H_{N}.$ Thus the finite-gap solution generated by $\Psi^{(N)}$ is parametrized by a hyperelliptic curve with $2N+1$ finite branch points. In addition to the branch point at infinity, this gives a genus of $N.$ In order to account for the additional branch point
at zero in the invariant hyperelliptic spectral curve, it is convenient to add the fixed root at $\lambda=0$ to the definition of the invariant polynomial $\mathscr{R}(\lambda).$

\begin{theorem}[Global Solution of the Focusing mKdV Equation] \label{globalfocusing}
Let $y_j=f_{j-1}, y_{N+j}= g_{j-1}, y_{2N+j} = h_{j-1},$ for $j= 1, \ldots, N,$ be a compatible local  solution of the $N$  autonomous ordinary differential equations~(\ref{Fode})
with initial data that satisfy  the focusing invariant reality conditions~(\ref{focusingreality}). 
Then the solution exists for all $t_k \in \mathbb{R},$ for $k=0, \ldots, N-1,$ and is uniformly bounded on $\mathbb{R}^N.$  In particular, the corresponding solution $u$ of the focusing mKdV equation~(\ref{mkdv}) is a uniformly bounded global solution. 
\end{theorem}

\begin{proof}
The compatibility of the hierarchy flows guarantees the existence of the smooth local  solution. We will show that  the invariant polynomial \(\mathscr R\) controls the Laurent-polynomial coefficients and yields uniform bounds on all dynamical variables.

The focusing reality conditions~(\ref{focusingreality}) imply that $b^2=-1$ and, for $\lambda \in \mathbb{R},$ $F_{N-1} (\lambda) = - \overline{F_{N-1} (\lambda)}$ and $H_{N} (\lambda) = - \overline{G_{N} (\lambda)},$ so that
\begin{equation} \label{Rfocusing}
\mathscr{R} (\lambda) = - \lambda^2 |F_{N-1}(\lambda)|^2 - \lambda |G_{N}(\lambda)|^2.
\end{equation}
Thus, for $\lambda \in \mathbb{R},$ 
\begin{equation}
\overline{\mathscr{R} (\lambda)} = \mathscr{R} (\lambda).
\end{equation}
Therefore, the coefficients of $\mathscr{R}$ are real and the roots of $\mathscr{R}(\lambda)=0$ are either real or occur in complex-conjugate pairs.

Now suppose that $\lambda \in \mathbb{R}$ and $\lambda > 0$ is a positive real root of $\mathscr{R}(\lambda)=0.$ Then equation~(\ref{Rfocusing}) implies that $\lambda$ is 
a zero of both $F_{N-1} (\lambda)$ and $G_{N}(\lambda),$ which means that $\lambda$ is also a double root of $\mathscr{R}(\lambda)=0,$ which is not allowed by the assumption that the roots are distinct. Thus $\mathscr{R}(\lambda)=0$ does not have any positive roots. All the nonzero roots are either negative or  complex-conjugate pairs. This means that there is an even number of negative roots. Therefore, 
\begin{equation}
\mathscr{R}(\lambda) =- \lambda P^+(\lambda) P^- (\lambda) Q (\lambda),
\end{equation}
where $P^{+}, P^{-}, Q$ are monic polynomials,  the roots of $P^{+}(\lambda)$ are all in the upper half plane, $P^-(\lambda) = \overline{P^+ (\overline{\lambda})}$, and $Q(\lambda)$ has an even number of negative roots and no other roots. Therefore, for $\lambda >0,$
\begin{equation}
\lambda^2 |F_{N-1}(\lambda)|^2 + \lambda |G_{N}(\lambda)|^2 =  \lambda |P^+(\lambda)|^2  Q (\lambda)
\end{equation}
implies 
\begin{equation}
\label{Fbound}
\lambda |F_{N-1} (\lambda)| \leq \sqrt{\lambda Q(\lambda)} |P^+ (\lambda)|
\end{equation}
and
\begin{equation}
\label{Gbound}
|G_{N} (\lambda)| \leq  \sqrt{Q(\lambda)} |P^+(\lambda)|.
\end{equation}
Note that $Q(\lambda) >0,$ since $\lambda >0$ and $Q(\lambda)$ is monic and has only negative roots. 
Equations~(\ref{Fbound}) and~(\ref{Gbound}) imply uniform bounds of the polynomials $|F_{N-1}(\lambda)|$ and $|G_{N}(\lambda)|$ on a closed finite interval of positive real numbers, e.g., on $1 \leq \lambda \leq 3.$ After an affine rescaling of the interval \([1,3]\) onto \([-1,1]\), standard coefficient estimates for bounded polynomials (e.g. Markov inequalities~\cite{mark 90}) imply uniform bounds on all coefficients of \(F_{N-1}(\lambda)\) and \(G_N(\lambda)\). The focusing reality conditions imply identical bounds for the coefficients of \(H_N(\lambda)\).
This, in turn, implies the uniform boundedness of all the dynamical variables that appear as coefficients in the polynomials $F_{N-1}, G_{N},$ and $H_{N}.$
Since solutions to polynomial vector fields can be continued as long as the solution remains bounded, the uniform coefficient bounds allow the local solution to be continued globally.
 In particular,  the finite-gap solution of the focusing mKdV equation, $u ({\bf t}) = -i f_{N-1} ({\bf t}),$ is uniformly bounded and exists for all times ${\bf t} \in \mathbb{R}^{N}.$

\end{proof}

\begin{theorem}[Global Solution of the Defocusing mKdV Equation] 
Let $y_j=f_{j-1}, y_{N+j}= g_{j-1}, y_{2N+j} = h_{j-1},$ for $j= 1, \ldots, N,$ be a compatible local  solution of the $N$  autonomous ordinary differential equations~(\ref{Fode}). 
Suppose, also, that the initial conditions are such that
\begin{enumerate}
\item[(a)] the reality conditions of the defocusing mKdV equation~(\ref{defocusingreality}) are satisfied, viz., $b=1$ and
\begin{equation}
\overline{F_{N-1}(\overline{\lambda})} = F(\lambda), \qquad \overline{G_{N}(\overline{\lambda})} = G(\lambda), \qquad \overline{H_{N}(\overline{\lambda})} = H(\lambda),
\end{equation} 
\item[(b)]  all of the nonzero roots of the invariant polynomial $\mathscr{R}(\lambda)$ are distinct negative numbers  ordered so that $\lambda_{k+1} < \lambda_{k} < 0,$ for $k=1, \ldots, 2N-1,$ and
\item[(c)] the degree $N$ polynomials $G_{N}(\lambda)$ and $H_{N}(\lambda)$ each have exactly one  negative root in each of the $N$ intervals $[\lambda_{2j},\lambda_{2j-1}],$ for $j= 1, \ldots, N.$ 
\end{enumerate}
Then the solution is uniformly bounded and exists for all $t_k \in \mathbb{R},$ for $k=0, \ldots, N-1.$  In particular, the corresponding solution $u$ of the defocusing mKdV equation~(\ref{mkdv})  is a uniformly bounded global solution. 
\label{globaldefocusing}
\end{theorem}

\begin{proof}
The defocusing reality conditions imply that the polynomials $F_{N-1}, G_{N},$ and $H_{N}$ have real coefficients and, therefore, their roots are either real or come in complex-conjugate pairs. The same symmetry is true for $\mathscr{R}(\lambda) = \lambda^2 |F_{N-1} (\lambda)|^2 + \lambda G_{N}(\lambda) H_{N} (\lambda).$ 

Now suppose that $\lambda \in \mathbb{R}$ is a  root of $G_{N}(\lambda)$ or $H_{N}(\lambda)$, then
\begin{equation}\label{rneg}
\mathscr{R} (\lambda) = \lambda \prod\limits_{j=1}^{2N} (\lambda - \lambda_{j} ) = \lambda^2 |F_{N-1}(\lambda)|^2 \geq 0.
\end{equation}
Thus $\lambda$ must be a root of $\mathscr{R}(\lambda)$ or less than an even number of the $2N+1$ nonpositive roots of $\mathscr{R} (\lambda)=0,$ by assumption (b).
Therefore $\lambda$ must lie in the interval $\lambda_{2j} \leq \lambda \leq \lambda_{2j-1},$ for some $j= 1, \ldots, N,$ because of the fixed root at zero. Moreover, the root 
$\lambda$ cannot leave the interval it is in and remain real since this would violate the inequality~(\ref{rneg}). However, by  assumption (c), initially there is exactly one root of $G_{N}(\lambda)$ and exactly one root of $H_{N}(\lambda)$ in each interval
$[\lambda_{2j}, \lambda_{2j-1}],$ for $j=1, \ldots, N.$  Since the coefficients of both $G_{N}(\lambda)$ and $H_{N}(\lambda)$ are real, their roots are either real or occur in complex-conjugate pairs. The intervals
$[\lambda_{2j},\lambda_{2j-1}],$ for $j=1,\ldots, N,$  are disjoint, so it is impossible for two roots of $G_{N}(\lambda)$ (or $H_{N}(\lambda)$) to flow continuously, collide, and become a complex-conjugate pair, because initially only one is in each of the disjoint real intervals to which they are individually constrained. Hence, the roots remain in these bounded intervals for all $t_k \in \mathbb{R},$ for $k=0, \ldots, N-1,$ for which the local solution continues to exist.  Thus the roots of $G_{N}(\lambda)$ and $H_{N}(\lambda)$ and, hence, the coefficients of $G_{N}(\lambda)$ and $H_{N}(\lambda)$ are uniformly bounded.

Moreover, since the coefficients of $G_N$ and $H_N$ are uniformly bounded as functions of ${\bf t} \in \mathbb{R}^{N},$
\begin{equation}
|\lambda|^2 |F_{N-1}(\lambda)|^2 \leq |\lambda| |G_{N}(\lambda)H_{N}(\lambda)| + |\mathscr{R}(\lambda)|
\end{equation}
implies that there exists a finite positive number $M$ such that
\begin{equation}
\max\limits_{1 \leq \lambda \leq 3} |F_{N-1} (\lambda)| \leq M,
\end{equation}
where the interval $[1,3]$ is chosen to avoid the value $\lambda=0,$ but otherwise could be any closed interval between two positive numbers.
Using an affine transformation of the interval $[1,3]$ to $[-1,1],$ we can apply Markov's inequality to conclude that the coefficients of $F_{N-1}(\lambda)$ are also uniformly bounded as functions of ${\bf t} \in \mathbb{R}^{N}$.

Since  the uniformly bounded coefficients of $F_{N-1}, G_{N},$ and $H_{N}$ are the components of the  compatible local solution of the commuting polynomial vector fields of equation~(\ref{Fode}), the local solution can be extended globally.   In particular,  the finite-gap solution of the defocusing mKdV equation, $u ({\bf t}) = f_{N-1} ({\bf t}),$ is uniformly bounded and exists for all times ${\bf t} \in \mathbb{R}^{N}.$
\end{proof}

\section{Maximal Amplitudes}

\begin{theorem}[Critical Points]\label{criticalpoints}
Let \(y_j=f_{j-1},\; y_{N+j}=g_{j-1},\; y_{2N+j}=h_{j-1}\), for \(j=1,\ldots,N\), be a compatible solution of the \(N\) autonomous ordinary differential equations~\eqref{Fode} under either the focusing or defocusing reality conditions.
If the point \(\mathbf t \in\mathbb R^N\) is a critical point of \(-i f_{N-1}:\mathbb R^N\to\mathbb R\) in the focusing case, or \( f_{N-1}:\mathbb R^N\to\mathbb R\) in the defocusing case, then
\begin{equation}
g_j(\mathbf t)=h_j(\mathbf t),
\qquad j=0,\ldots,N-1,
\end{equation}
at \((t_0,t_1,\ldots,t_{N-1})=\mathbf t\).
\end{theorem}

\begin{proof}
Since \(\mathbf t\) is a critical point,
\begin{equation}
\partial_k f_{N-1}(\mathbf t)=0,
\qquad k=0,\ldots,N-1.
\end{equation}
Therefore, equation~(\ref{leadingA}) implies
\begin{equation}
g_{N-k-1}(\mathbf t)=h_{N-k-1}(\mathbf t),
\qquad k=0,\ldots,N-1.
\end{equation}
Hence,
\begin{equation}
g_j(\mathbf t)=h_j(\mathbf t),
\qquad j=0,\ldots,N-1.
\end{equation}

\end{proof}

\begin{definition} \label{squareroots}
Let
\begin{equation}
\mathscr R(\lambda)
=
b^2\lambda
\prod_{j=1}^{2N}
(\lambda-\lambda_j)
\end{equation}
be the invariant polynomial of Lemma~\ref{invariant}, where
$\lambda_1,\ldots,\lambda_{2N},$
denote its distinct nonzero roots.  Assume there are no positive real roots.
For each $j=1, \ldots, 2N,$ define $E_j \in \mathbb{C}$ to be the unique square root of $\lambda_j$ with positive imaginary part, viz.,  $\lambda_j = E_j^2$ and $\Im \left( E_j \right) > 0,$ and
define
\begin{equation}
\mathcal E^+ = \{ E_1, \ldots, E_{2N} \}.
\end{equation}
\end{definition}

\begin{theorem}[Sharp Upper Bound for the  Focusing mKdV Equation] \label{sharpfocusing}
Let
\[
y_j=f_{j-1}, \qquad
y_{N+j}=g_{j-1}, \qquad
y_{2N+j}=h_{j-1},
\qquad j=1,\ldots,N,
\]
be a compatible global solution of the \(N\) autonomous ordinary differential equations~\eqref{Fode}, with initial conditions satisfying the focusing reality conditions given in Theorem~\ref{globalfocusing}.
Then the corresponding solution \(u\) of the mKdV equation~(\ref{mkdv}) satisfies
\begin{equation}
|u(x,t)|
\le
\sum_{E\in\mathcal E^+}\Im \left( E \right).
\end{equation}
Moreover, the bound is sharp, in the sense that there exists a solution which attains the upper bound.
\end{theorem}

\begin{proof}
Fix \(\mathscr R(\lambda)\) the invariant polynomial of the given global solution. 
Let \(\mathcal S\) denote the set of all points in the $3N$-dimensional phase space of the commuting hierarchy which (i)  satisfy the focusing reality conditions and (ii)  satisfy the defining equation~(\ref{Reqn}) of the invariant polynomial \(\mathscr R(\lambda)\). Every point of $\mathcal S$ serves as initial data for a unique solution of the commuting hierarchy which preserves the reality conditions and the invariant polynomial $\mathscr{R} (\lambda),$ therefore every point of $\mathcal S$ lies on the orbit of a real solution with invariant polynomial $\mathscr{R} (\lambda).$ Conversely, the orbit of any real solution with invariant polynomial $\mathscr{R} (\lambda)$ must lie entirely on $\mathcal S,$ because the reality conditions and $\mathscr{R}(\lambda)$ are invariants of the orbit. Hence $\mathcal S$ coincides with the union of the orbits of all solutions satisfying the focusing reality conditions and having invariant polynomial $\mathscr{R}(\lambda).$ Equivalently, $\mathcal S$ coincides with the set of the initial data of all solutions satisfying the focusing reality conditions and having invariant polynomial  $\mathscr{R}(\lambda).$
The invariant polynomial condition is the common zero  set of a finite set of polynomials of the dynamical variables, and the reality conditions define a closed subspace of the phase space. Hence \(\mathcal S\) is closed. 

Moreover, the construction of the uniform  bounds in Theorem~\ref{globalfocusing}  shows that the bounds on the dynamical variables along a real orbit depend only on the invariant polynomial $\mathscr{R}(\lambda)$. Thus the bounds  are identical for every real orbit having  the given invariant polynomial $\mathscr{R}$. Since the set $\mathcal S$ is identical to the set of points on these uniformly bounded orbits, $\mathcal S$ is bounded.   Therefore, \(\mathcal S\) is a closed  bounded subset of the finite-dimensional phase space. Therefore, \(\mathcal S\) is compact. 

Consider the projection of \(\mathcal S\) onto the \(f_{N-1}\)-coordinate of the finite-dimensional phase space. Since this projection is compact, there exists a point ${\bf s}^{\max}  \in \mathcal S$ maximizing $|f_{N-1}|.$ Note that ${\bf s}^{\max}$ is not necessarily unique, but the maximum value of $|f_{N-1}|$ for points in $\mathcal S$ is unique.  Define
${\bf y}^{\max} ({\bf t})$ to be the solution of the hierarchy flows with initial data ${\bf y}^{\max} ({\bf 0}) = {\bf s}^{\max}.$ Let $f_{N-1}^{\max} ({\bf t})$ be the
$f_{N-1}$-coordinate of ${\bf y}^{\max} ({\bf t}).$ And let $u^{\max} ({\bf t}) =- i f_{N-1}^{\max} ({\bf t})$ be the global finite-gap solution of the focusing mKdV equation coming from the bounded global solution ${\bf y}^{\max} ({\bf t})$ of the hierarchy flows. Since every point of the orbit of ${\bf y}^{\max}( {\bf t})$ belongs to $\mathcal S,$ and the initial point was chosen to maximize $|f_{N-1}|$ on $\mathcal{S},$ it follows that
\begin{equation}
|u^{\max} ({\bf t} ) | \leq |u^{\max} ({\bf 0})|,
\end{equation}
for all $\bf{t} \in \mathbb{R}^{N}.$ Without loss of generality, we can assume that $u^{\max} ({\bf 0}) \geq 0,$ since the mKdV equation~(\ref{mkdv}) is invariant under the transformation $u \mapsto -u.$
Therefore, \(\mathbf t=\mathbf 0\) is the location of a global maximum and a critical point of $u^{\max}: \mathbb{R}^{N} \rightarrow \mathbb{R}$ defined by $u^{\max}({\bf t}) = -i f_{N-1}^{\max} ({\bf t}),$ with respect to all $N$ flows in the hierarchy.

Since ${\bf t} = {\bf 0}$ is a critical point of $u^{\max}({\bf t}) = -i f^{\max} ({\bf t}),$ Theorem~\ref{criticalpoints} and $g_N = h_N = b= i$ imply that
$G_N(\lambda)=H_N(\lambda)$ at ${\bf t} = {\bf 0}.$ Therefore, at $\bf{t} = \bf{0},$ the focusing reality conditions imply that the coefficients of $G_{N} (\lambda)$ are purely imaginary, viz.,
\begin{equation}
G_N(\lambda)= H_{N} (\lambda) = - \overline{ G ( \overline{\lambda} ) },
\end{equation}
and at ${\bf t} ={\bf 0}$ the invariant polynomial can be written as
\begin{equation}
\mathscr{R}(\lambda) = \lambda^2 F_{N-1}^2(\lambda) + \lambda G_{N}^2 (\lambda) = - \lambda \prod\limits_{j=1}^{2N} (\lambda - \lambda_j).
\end{equation}
If we define $\lambda=E^2,$ then we can factor this equation. 
By Definition~\ref{squareroots}, the set $\mathcal E^+ = \{E_1, \ldots, E_{2N}\}$ consists of the  square roots in the upper-half plane of the invariant roots $\lambda_j$ (recall that, in Theorem~\ref{globalfocusing}, we showed that all of the invariant roots are nonpositive). These square roots are necessarily distinct because the invariant roots themselves are distinct.
After canceling out the root at $\lambda=0,$ we obtain
\begin{equation}\label{reducedRfocusing}
E^2F_{N-1}^2(E^2) +G_N^2(E^2) =-\prod\limits_{j=1}^{2N}(E^2-E_j^2).
\end{equation}

Equation~(\ref{reducedRfocusing}) implies the factorization
\begin{equation}\label{factorfirst}
\left(E F_{N-1} (E^2) + i G_{N} (E^2) \right) \left( E F_{N-1} (E^2) - i G_{N} (E^2) \right) = -\prod\limits_{j=1}^{2N}(E^2-E_j^2).
\end{equation}
Suppose that $(E^2-E_j^2)$ were a factor of the first (or second) factor on the left-hand side of equation~(\ref{factorfirst}). That would imply that $E= \pm E_j$ were both roots of that factor and, hence, that $F(E_j^2)=G(E_j^2) =0.$ Equation~(\ref{reducedRfocusing}) would then force $E=E_j^2$ to be a double root of the invariant polynomial, which is not allowed.
Therefore, using  $g_N = h_N =b= i,$ the factorization of equation~(\ref{factorfirst}) must have the form, 
\begin{subequations}\label{factorAA}
\begin{align}
E F_{N-1} (E^2) + i  G_N(E^2) &=   -\prod\limits_{j=1}^{2N} (E - \epsilon_j E_j) \label{factorA} \\
E F_{N-1} (E^2) - i  G_N (E^2) &=   \prod\limits_{j=1}^{2N} (E +\epsilon_j E_j),\label{factorB}
\end{align}
\end{subequations}
where $\epsilon_j = \pm 1,$ for $j=1, \ldots, 2N.$ The construction guarantees that the polynomials $F_{N-1}(E^2)$ and $G_{N}(E^2)$ are well-defined by equations~(\ref{factorA}) and~(\ref{factorB}), because the factorization must have the given form. It  can be verified directly that the symmetric polynomials of even and odd degrees of the roots on each side of each equation give consistent expressions for $F_{N-1} (E^2)$ and $G_{N}(E^2)$.

Consider the coefficient of $E^{2N-1}$ in equation~(\ref{factorA}) or~(\ref{factorB}),
\begin{equation}\label{leadingsum}
f_{N-1}^{\max}({\bf 0}) = \sum\limits_{j=1}^{2N} \epsilon_j E_j,
\end{equation}
which must be purely imaginary. The maximal upper bound occurs when $\epsilon_{j} = 1,$ for $j=1, \ldots, 2N,$ viz.,
\begin{equation}
|f_{N-1}^{\max}({\bf 0})|  = | \Im \left( f_{N-1}^{\max}({\bf 0}) \right) |=  \sum\limits_{j=1}^{2N} \Im \left( E_j \right), \label{max1}
\end{equation}
provided that this factorization corresponds to initial data satisfying the focusing reality conditions.
To check that the focusing reality conditions are satisfied, solve equation~(\ref{factorAA}),
\begin{subequations}\label{solvedfactors}
\begin{align}
F_{N-1}(E^2) &= \frac{1}{2E} \left(\prod\limits_{j=1}^{2N} (E + E_j) - \prod\limits_{j=1}^{2N} (E -  E_j) \right) \\
G_{N} (E^2) &= \frac{i}{2} \left( \prod\limits_{j=1}^{2N} (E +  E_j) + \prod\limits_{j=1}^{2N} (E - E_j) \right).
\end{align}
\end{subequations}
The roots $E_j \in \mathcal E^+$ are either purely imaginary or occur in pairs $E_{j}, - \overline{E_{j}}.$ If we label the elements of $\mathcal E^+$ as $E_{2k-1}=-\overline{E_{2k}},$ for $k=1, \ldots, M,$ and $E_{k}=-\overline{E_k},$ for $k=2M+1, \ldots, 2N,$ then 
\begin{subequations}
\begin{align}
\prod_{j=1}^{2N} (E-  E_j) &= \prod_{k=1}^{M} (E^2 - (E_{2k} - \overline{E_{2k}})E - |E_{2k}|^2) \prod_{k=2M+1}^{2N} (E -  E_k), \\
\prod_{j=1}^{2N} (E+  E_j) &= \prod_{k=1}^{M} (E^2 + (E_{2k} - \overline{E_{2k}})E - |E_{2k}|^2) \prod_{k=2M+1}^{2N} (E + E_k). 
\end{align}
\end{subequations}
Therefore,
\begin{equation}
\overline{\prod_{j=1}^{2N} (\overline{E}-  E_j) }  = \prod_{j=1}^{2N} (E+  E_j).
\end{equation}
And, finally, $F_{N-1} (E^2) = -\overline{F_{N-1} (\overline{E}^2)}$ and $H_{N}(E^2)=G_{N} (E^2) = -\overline{G_{N}(\overline{E}^2)},$ showing that the focusing reality conditions are satisfied at this point in the phase space and, so, the initial data defined by the factorization in equation~(\ref{factorAA}), with $\epsilon_j = 1,$ for all $j=1, \ldots, 2N,$ lies in $\mathcal S.$

Therefore, equation~(\ref{max1}) implies that the solution $u^{\max} ({\bf t})$ satisfies
\begin{equation}
|u^{\max} ({\bf 0})| = \sum\limits_{j=1}^{2N} \Im \left( E_j \right).
\end{equation}
Hence
\begin{equation}
|u(x,t)|
\le
\sum_{E \in\mathcal E^+}\Im \left(E \right),
\end{equation}
and the bound is sharp because it is attained by the maximizing solution.
\end{proof}

\begin{remark}
Previous derivations~\cite{wrig 19, wrig 20, wrig 24} of sharp amplitude bounds for finite-gap solutions relied on global compactness of individual trajectories on invariant tori in order to guarantee that each solution attained its maximum. The present argument avoids this requirement entirely. Instead, the sharp bound is obtained by maximizing directly over the compact set of admissible initial data sharing the same invariant polynomial. Thus, only the existence of a single maximizing solution is required.
\end{remark}

\begin{remark} The proof of Theorem~\ref{sharpfocusing} uses only the invariant polynomial and the finite-dimensional hierarchy structure. The sine-Gordon equation arises as the first negative flow of the combined sG/mKdV hierarchy~\cite{gesz 03, gesz 00}, but under a different hierarchy reduction than the mKdV equation. The existence of analogous sharp bounds for finite-gap sine-Gordon solutions will be investigated elsewhere.
\end{remark}

\begin{theorem}[Sharp Upper Bound for the  Defocusing mKdV Equation]\label{sharpdefocusing}
Let
\[
y_j=f_{j-1}, \qquad
y_{N+j}=g_{j-1}, \qquad
y_{2N+j}=h_{j-1},
\qquad j=1,\ldots,N,
\]
be a compatible global solution of the \(N\) autonomous ordinary differential equations~\eqref{Fode}, with initial conditions satisfying the defocusing reality conditions given in Theorem~\ref{globaldefocusing}.
Then the corresponding solution \(u\) of the defocusing mKdV equation~(\ref{mkdv}) satisfies
\begin{equation}
|u(x,t)|
\le
\sum_{j=1}^{N} \Im \left( E_{2j} -  E_{2j-1} \right),
\end{equation}
where the elements of $\mathcal E^+$ are purely imaginary and are ordered by $0 < \Im \left( E_1 \right) < \Im \left( E_2 \right) < \ldots  < \Im \left( E_{2N} \right).$ 
Moreover, the bound is sharp, in the sense that there exists a solution which attains the upper bound.
\end{theorem}

\begin{proof}
Fix \(\mathscr R(\lambda)\) the invariant polynomial of the given global solution defined by Theorem~\ref{globaldefocusing}.
Let \(\mathcal S\) denote the set of all points in the $3N$-dimensional phase space of the commuting hierarchy which (i)  satisfy the defocusing reality conditions, (ii) satisfy the defining equation~(\ref{Reqn}) of the invariant polynomial \(\mathscr R(\lambda)\), and  (iii) define polynomials $G_{N}(\lambda)$ and $H_{N}(\lambda)$ which each have exactly one root in each of the intervals $[\lambda_{2j},\lambda_{2j-1}],$ for $j=1, \ldots, N.$  The condition that the degree $N$ polynomials $G_{N}(\lambda)$ and $H_{N}(\lambda)$ possess exactly one root in each closed interval $[\lambda_{2j},\lambda_{2j-1}]$ is preserved under limits because roots depend continuously on coefficients and cannot leave their respective intervals because of the reality conditions and equation~(\ref{rneg}).
 Thus,  $\mathcal S$ is a closed subset of the finite-dimensional phase space, just as in the proof of Theorem~\ref{sharpfocusing}.

Every point of $\mathcal S$ serves as initial data for a unique solution of the commuting hierarchy which preserves the conditions of Theorem~\ref{globaldefocusing} and the invariant polynomial $\mathscr{R} (\lambda),$ therefore every point of $\mathcal S$ lies on the orbit of a solution satisfying the conditions of Theorem~\ref{globaldefocusing} with invariant polynomial $\mathscr{R} (\lambda).$ Conversely, the orbit of any solution satisfying the conditions of Theorem~\ref{globaldefocusing} with invariant polynomial $\mathscr{R} (\lambda)$ must lie entirely on $\mathcal S,$ because the conditions of the theorem and $\mathscr{R}(\lambda)$ are invariants of the orbit. Hence $\mathcal S$ coincides with the union of the orbits of all solutions satisfying the conditions of Theorem~\ref{globaldefocusing} and having invariant polynomial $\mathscr{R}(\lambda).$ 
Thus, the uniform bound on the global solutions constructed in Theorem~\ref{globaldefocusing} applies to the orbits which constitute the set $\mathcal S.$ Therefore, $\mathcal S$ is a closed and bounded subset of the finite-dimensional phase space. Therefore, $\mathcal S$ is compact.

Consider the projection of \(\mathcal S\) onto the \(f_{N-1}\)-coordinate of the finite-dimensional phase space. Since this projection is compact, there exists a point ${\bf s}^{\max}  \in \mathcal S$ maximizing $|f_{N-1}|.$ Note that ${\bf s}^{\max}$ is not necessarily unique, but the maximum value of $|f_{N-1}|$ for points in $\mathcal S$ is unique.  Define ${\bf y}^{\max} ({\bf t})$ to be the solution of the hierarchy flows with initial data ${\bf y}^{\max} ({\bf 0}) = {\bf s}^{\max}.$ Let $f_{N-1}^{\max} ({\bf t})$ be the
$f_{N-1}$-coordinate of ${\bf y}^{\max} ({\bf t}).$ And let $u^{\max} ({\bf t}) = f_{N-1}^{\max} ({\bf t})$ be the global finite-gap solution of the defocusing mKdV equation coming from the bounded global solution ${\bf y}^{\max} ({\bf t})$ of the hierarchy flows. Since every point of the orbit of ${\bf y}^{\max}( {\bf t})$ belongs to $\mathcal S,$ and the initial point was chosen to maximize $|f_{N-1}|$ on $\mathcal{S},$ it follows that
\begin{equation}
|u^{\max} ({\bf t} ) | \leq |u^{\max} ({\bf 0})|,
\end{equation}
for all $\bf{t} \in \mathbb{R}^{N}.$ Without loss of generality, we can assume that $u^{\max} ({\bf 0}) \geq 0,$ since the mKdV equation~(\ref{mkdv}) is invariant under the transformation $u \mapsto -u.$
Therefore, \(\mathbf t=\mathbf 0\) is the location of a global maximum and a critical point of $u^{\max}: \mathbb{R}^{N} \rightarrow \mathbb{R}$ defined by $u^{\max}({\bf t}) =  f_{N-1}^{\max} ({\bf t}),$ with respect to all $N$ flows in the hierarchy.

Since ${\bf t} = {\bf 0}$ is a critical point of ${u}_{\max} ({\bf t}) =f^{\max}_{N-1} ({\bf t}),$ Theorem~\ref{criticalpoints} and $g_N = h_N =b= 1$ imply that $G_N(\lambda)=H_N(\lambda)$ at ${\bf t} = {\bf 0}.$ Therefore, at $\bf{t} = \bf{0},$ the invariant polynomial has the form
\begin{equation}\label{RR}
\mathscr{R}(\lambda) = \lambda^2 F_{N-1}^2(\lambda) + \lambda G_{N}^2 (\lambda) =  \lambda \prod\limits_{j=1}^{2N} (\lambda - \lambda_j).
\end{equation}
Under the defocusing reality conditions of Theorem~\ref{globaldefocusing}, the invariant polynomial has negative real roots
\begin{equation}
\lambda_{2N} = - \beta_{2N}^2, \lambda_{2N-1} = - \beta_{2N-1}^2, \ldots, \lambda_1 = - \beta_1^2,
\end{equation}
where $\beta_j >0,$ for $j=1, \ldots, 2N,$ and $\mathcal E^+ = \{i \beta_1, \ldots, i \beta_{2N}\}$ are the square roots in the upper-half plane. Setting $\lambda=E^2,$ the invariant polynomial in
equation~(\ref{RR}) reduces to 
\begin{equation}\label{reducedR}
E^2F_{N-1}^2(E^2) +G_N^2(E^2) =\prod\limits_{j=1}^{2N}(E^2+\beta_j^2).
\end{equation}
Notice that a factorization of the form
\begin{equation}
E F_{N-1} (E^2) + i G_{N}(E^2) = i (E^2 + \beta_k^2)\prod\limits_{j \neq k}^{2N} (E - i \epsilon_{j} \beta_j),
\end{equation}
where $\epsilon_{j} = \pm 1,$ for $j=1, \ldots, 2N,$  is not possible, because setting $E=i \epsilon_k \beta_k$ implies that
\begin{equation}\label{eqnep}
i \epsilon_k \beta_k F_{N-1} (-\beta_k^2) + i G_{N} (- \beta_k^2) = 0.
\end{equation}
Since $\epsilon_k = \pm 1,$ equation~(\ref{eqnep}) implies that $F_{N-1} (-\beta_k^2) = G_{N}(-\beta_k^2) = 0.$ In that case, equation~(\ref{reducedR}) shows that $\lambda= E^2= - \beta_k^2$ is a double root of the invariant polynomial $\mathscr{R}(\lambda)$, which is not permitted.

Therefore, the factorization of equation~(\ref{reducedR}) must have the form
\begin{subequations}\label{factor00}
\begin{align}
E F_{N-1} (E^2) + i G_{N}(E^2) &= i \prod\limits_{j =1}^{2N} (E - i \epsilon_{j} \beta_j), \label{factor11} \\
E F_{N-1} (E^2) - i G_{N}(E^2) &= -i\prod\limits_{j =1}^{2N} (E + i \epsilon_{j} \beta_j). \label{factor22}
\end{align}
\end{subequations}
By considering the symmetric polynomials of even and odd degrees on both sides of equations~(\ref{factor11}) and~(\ref{factor22}), we see that the coefficients of $F_{N-1}(E^2)$ and $G_{N}(E^2)$ are real, as required by the defocusing reality conditions.
However, we must still check that the roots of $G_{N}(\lambda)$ lie in the correct intervals.  

To determine the permissible signs $\epsilon_j = \pm 1$ in the factorization, we will use a deformation argument on the polynomials in equation~(\ref{factor00}). Allow the parameters $\beta_{2j-1}$ and $\beta_{2j}$ to change continuously in such a way that they remain distinct, but $\beta_{2j-1} \rightarrow \beta_{2j},$ for $j=1, \ldots, N.$ In other words, we collapse the $N$ gaps in the spectral configuration. Assume at the beginning of the deformation that the factorization~(\ref{factor00}) is permissible, viz., the roots of $G_{N}(\lambda)$ are positioned so that there is exactly one in each gap. During the deformation process, the coefficients and the roots of $G_{N}(\lambda)$ vary continuously, but the roots of $G_{N}(\lambda)$ remain trapped in the spectral gaps because if one of the roots were to hit an endpoint $\lambda_k = -\beta_k^2,$ then equation~(\ref{reducedR}) would imply that this root was also a root of $F_{N-1}(\lambda)$ and hence a double root of the invariant polynomial $\mathscr{R}(\lambda),$ which is impossible by assumption on $\mathscr{R}(\lambda).$

Passing to the limit of this deformation process, the roots of $G_{N}(\lambda)$ collapse onto $\lambda_j=-\beta_{2j-1}^2 = -\beta_{2j}^2,$ for $j=1, \ldots, N.$ Moreover, these collapsed root pairs are distinct from each other, viz., $\beta_{2j} \neq \beta_{2k},$ for all $j \neq k.$
Hence, in the limit, equation~(\ref{reducedR}) implies that
\begin{equation}
E^2F_{N-1}^2(E^2) +\prod\limits_{j=1}^{N} (E^2+\beta_{2j}^2)^2 =\prod\limits_{j=1}^{N}(E^2+\beta_{2j}^2)^2 \Rightarrow F_{N-1}(E^2) \equiv 0.
\end{equation}
So, in the limit of this deformation process, equation~(\ref{factor11}) and $F_{N-1} (E^2) \equiv 0$ imply
\begin{equation}\label{quadraticform}
  \prod\limits_{j=1}^{N} \left( E^2 + \beta_{2j}^2 \right) =  \prod\limits_{j=1}^{N} (E - i \epsilon_{2j-1} \beta_{2j})(E - i \epsilon_{2j} \beta_{2j}).
\end{equation}
Now suppose that for some value of $j,$ $\epsilon_{2j-1}=\epsilon_{2j}$ in equation~(\ref{quadraticform}). Then the right-hand side of equation~(\ref{quadraticform}) has a double root at $E= i \epsilon_{2j} \beta_{2j}.$ But the factor $E^2+\beta_{2j}^2$ on the left-hand side contributes roots $E=\pm i \beta_{2j},$ which cannot both be equal to $E=i \epsilon_{2j} \beta_{2j},$ since $\beta_{2j} >0.$ Therefore, $\beta_{2j} = \beta_{2k}$ for some  $k \neq j,$ which is not possible in our construction. Thus,  the  permissible factorizations satisfy $\epsilon_{2j-1} = - \epsilon_{2j}.$

Therefore, equating coefficients of $E^{2N-1}$ in the original factorization (before the deformation) given by equation~(\ref{factor11}) (or equation~(\ref{factor22})), we conclude that the maximal amplitude is given by
\begin{equation}\label{defocusingumax}
|u^{\max} ({\bf 0})| =|f_{N-1}^{\max}({\bf 0})| = \sum\limits_{j=1}^{N} \epsilon_{2j} (\beta_{2j} - \beta_{2j-1}).
\end{equation} 
for some choice of $\epsilon_{2j} = \pm 1,$ for $j=1, \ldots, N.$
The choice $\epsilon_{2j} =1,$ for $j=1, \ldots, N,$ maximizes the sum in equation~(\ref{defocusingumax}), because $\beta_{2j-1} < \beta_{2j}$. 

Therefore, to complete the proof,  it is sufficient to show that, if $\epsilon_{2j} = - \epsilon_{2j-1} = 1,$ for $j=1, \ldots, N,$  then the initial data for the coefficients of $F_{N-1}(\lambda)$ and  $G_{N}(\lambda)$ defined by the factorization in equation~(\ref{factor00}) produces a 
$G_{N}(\lambda)$ which has exactly one root in each of the $N$ finite gaps. Setting $\epsilon_{2j}=-\epsilon_{2j-1}=1$ in equation~(\ref{factor00}), define the polynomial
\begin{equation}
P(E) =E F_{N-1} (E^2) + i G_{N}(E^2) = i \prod\limits_{j =1}^{N} (E - i  \beta_{2j})(E+i \beta_{2j-1} ),
\end{equation}
where $0 < \beta_{1} < \beta_{2} < \ldots < \beta_{2N}.$ Then
\begin{equation}
G_{N}(E^2) = \frac{1}{2i} ( P(E) + P(-E) )
\end{equation}
because $G_{N}(E)$ is completely determined by the even powers of $P(E).$ 
Then, for $\eta \in \mathbb{R},$ define
\begin{equation}
T(\eta) = i P(i \eta) = \prod\limits_{j=1}^{N} (\eta - \beta_{2j})(\eta + \beta_{2j-1}),
\end{equation} 
so that $T(\eta)$ is an even polynomial with real coefficients and simple roots at 
\begin{equation}\label{signseq}
-\beta_{2N-1} <  \ldots < -\beta_{1} < \beta_{2} < \beta_{4} < \ldots < \beta_{2N}.
\end{equation}
In particular, the sign of $T(\eta)$ is positive as $\eta \rightarrow \pm \infty$ and the sign alternates between the consecutive roots.
Therefore,
\begin{equation}
G_{N}(-\eta^2) =-\frac{1}{2} (T(\eta) + T(-\eta)),
\end{equation}
and
\begin{equation}
G_{N} (-\beta_{2j}^2) G_{N} (-\beta^2_{2j-1}) =\frac{1}{4} T(-\beta_{2j}) T(\beta_{2j-1}) <0,
\end{equation}
because  $T(\beta_{2j})=T(-\beta_{2j-1})=0$ and
\begin{subequations}
\begin{align}
\operatorname{sgn} T(-\beta_{2j})
&=
(-1)^{N+j},\\
\operatorname{sgn} T(\beta_{2j-1})
&=
(-1)^{N+j+1},
\end{align}
\end{subequations}
for $j=1, \ldots, N.$
Thus,  $G_{N} (\lambda)$ changes sign at the endpoints of each interval $(-\beta^2_{2j},  -\beta^2_{2j-1}),$ for $j=1, \ldots, N,$ and the intermediate value theorem implies that $G_{N}(\lambda)$ has at least one root in each of the $N$ finite gaps. Since $G_{N}(\lambda)$ has degree $N$, we conclude that there
is exactly one root of $G_{N}(\lambda)$ in each of the required intervals with this initial data. 

Therefore, initial data specified in the factorization of equation~(\ref{factor00}) with $\epsilon_{2j}=-\epsilon_{2j-1}=1,$ for $j=1, \ldots, N,$ satisfy all the conditions of Theorem~\ref{globaldefocusing}. The corresponding solution  of the defocusing mKdV equation~(\ref{mkdv}) is the maximizing solution $u^{\max} ({\bf t}),$ which
satisfies equation~(\ref{defocusingumax}).
Hence
\begin{equation}
|u(x,t)|
\le
\sum\limits_{j=1}^{N} \left( \beta_{2j} - \beta_{2j-1} \right) = \sum\limits_{j=1}^{N} \Im \left( E_{2j} -  E_{2j-1} \right),
\end{equation}
and the bound is sharp because it is attained by the maximizing solution.
\end{proof}

\section{Soliton Examples}
\subsection{Focusing Case}
For the focusing mKdV equation~(\ref{mkdv}), the $N=1$ polynomial ansatz is
\begin{subequations}\label{genus1ansatz}
\begin{align}
F_{0} (\lambda) &=  i u, \\
G_{1} (\lambda) &= i (\lambda - \mu), \\
H_{1} (\lambda) &= i (\lambda - \overline{\mu}),
\end{align}
\end{subequations}
where $u \in \mathbb{R}$ and $\mu \in \mathbb{C}.$
The invariant spectral curve is
\begin{equation}\label{focusingcurvesoliton}
\lambda^2 u^2 + \lambda (\lambda - \mu) (\lambda - \overline{\mu}) = \lambda (\lambda - \lambda_1) (\lambda - \lambda_2).
\end{equation}
Assuming $\lambda_1 \neq \lambda_2,$ substitution of $\lambda=\lambda_1$ or $\lambda=\lambda_2$ into equation~(\ref{focusingcurvesoliton}) shows that neither $\lambda_1$ nor $\lambda_2$ can be a positive real number. Assuming that $\lambda_1$ and  $\lambda_2$ are distinct and nonzero, they are either two distinct negative numbers or a complex-conjugate pair. We can use equation~(\ref{focusingcurvesoliton}) to obtain a quadratic equation for $\mu.$ The explicit solution is
\begin{equation}
\mu = \frac{1}{2} \left( u^2 + s_1 \pm i \sqrt{4 s_2 - (u^2+s_1)^2} \right),
\end{equation}
where $s_1 = \lambda_1 +\lambda_2$ and $s_2 = \lambda_1 \lambda_2.$ The focusing reality condition requires that the discriminant of the quadratic polynomial in $\mu$ be nonpositive so that $\mu$ exists in a complex-conjugate pair, i.e.,
\begin{equation}
(u^2+s_1)^2 - 4 s_2 \leq 0  
\end{equation}
or, equivalently,
\begin{equation}
-s_1 - 2 \sqrt{s_2} \leq u^2 \leq -s_1 + 2 \sqrt{s_2}.
\end{equation}
Using $\lambda_1 = E_1^2$ and $\lambda_2 = E_2^2,$ where either $E_1 = i \beta, E_2 = i \gamma,$ or $E_1 = \alpha + i \beta, E_2 = -\alpha + i \beta,$ and $\beta, \gamma >0,$ we obtain $s_1 = \lambda_1 +\lambda_2=E_1^2 + \overline{E_1}^2$ and $s_2 = |E_1|^4.$ 
Thus,
\begin{equation}
u^2 \leq -E_1^2 -\overline{E_1}^2 + 2|E_1|^2 = - \left(E_1 - \overline{E_1} \right)^2.
\end{equation} 
Therefore, $|u| \leq \beta + \gamma$ (for two negative invariant roots) or $|u| \leq 2 \beta$ (for a complex-conjugate pair of invariant roots), as predicted by Theorem~\ref{sharpfocusing}.

The dynamical equation for $u$ is obtained by substituting equation~(\ref{genus1ansatz}) into the $N=1$ form of the $x$-flow  equation~(\ref{genus0xeqnsA}), viz.,
\begin{equation}\label{umu}
u_x = -i \left(\mu - \overline{\mu} \right).
\end{equation} 
In the degenerate case $E_1 = E_2 = i \beta,$ solving equation~(\ref{umu}) produces the traveling-wave profile of the soliton solution of the focusing mKdV equation~(\ref{mkdv}), 
\begin{equation}
u(x) = 2 \beta \sech \left( 2 \beta (x-x_0) \right).
\end{equation}

In the degenerate limit, the one-phase elliptic solution degenerates into the classical one-soliton, which is a homoclinic orbit connecting the equilibrium $u=0$ to itself. The upper turning point of the elliptic solution persists in the limit as the peak of the homoclinic orbit. Consequently, the soliton continues to attain the same sharp amplitude $2 \beta$ predicted by Theorem~\ref{sharpfocusing}.

\subsection{Defocusing Case}
For the defocusing mKdV equation, the $N=1$ polynomial ansatz is
\begin{subequations}
\begin{align}
F_{0} (\lambda) &=   u, \\
G_{1} (\lambda) &=  (\lambda - \mu_1), \\
H_{1} (\lambda) &=  (\lambda - \mu_2),
\end{align}
\end{subequations}
where $u, \mu_1, \mu_2 \in \mathbb{R}.$
The invariant spectral polynomial is
\begin{equation}\label{defocusingcurvesoliton}
\lambda^2 u^2 + \lambda (\lambda - \mu_1) (\lambda - \mu_2) = \lambda (\lambda - \lambda_1) (\lambda - \lambda_2).
\end{equation}
We can use equation~(\ref{defocusingcurvesoliton}) to obtain a quadratic equation with roots $\mu=\mu_1, \mu_2.$ The explicit solution is
\begin{equation}
\mu = \frac{1}{2} \left( u^2 + s_1 \pm  \sqrt{(u^2+s_1)^2-4 s_2} \right),
\end{equation}
where $s_1 = \lambda_1 +\lambda_2$ and $s_2 = \lambda_1 \lambda_2.$ The defocusing reality condition requires  that the discriminant of the quadratic polynomial be nonnegative so that $\mu \in \mathbb{R},$ viz., 
\begin{equation}
(u^2+s_1)^2 - 4 s_2 \geq 0, 
\end{equation}
which is equivalent, for bounded solutions, to
\begin{equation}\label{u2}
u^2 \leq -s_1 - 2 \sqrt{s_2}.
\end{equation}
To find bounded solutions satisfying equation~(\ref{u2}), set $\lambda_2 = -\beta_2^2$ and $\lambda_1=-\beta_1^2,$ where $0 < \beta_1 < \beta_2,$ so that
\begin{equation}
u^2 \leq (\beta_2-\beta_1)^2.
\end{equation}
Now choose any real initial $u$ such that $0< u^2 \leq (\beta_2-\beta_1)^2,$ we will show that the initial data satisfies the root conditions of Theorem~\ref{globaldefocusing}. The roots we need to check are $\mu_1$ and $\mu_2$ which satisfy the quadratic equation,
\begin{equation}
Q(\mu) = \mu^2- (u^2-\beta_1^2-\beta_2^2) \mu + \beta_1^2\beta_2^2 =0.
\end{equation}  
Now $Q(-\beta_1^2) = \beta_1^2 u^2 > 0$ and $Q(-\beta_2^2) = \beta_2^2 u^2 > 0.$ The minimum of $Q(\mu)$ occurs at 
\begin{equation}
\mu^* = \frac{1}{2} (u^2 - \beta_1^2 - \beta_2^2) > -\frac{1}{2} (\beta_1^2 + \beta_2^2) > -\beta_2^2 .
\end{equation}
Also,
\begin{equation}
\mu^* = \frac{1}{2} (u^2 - \beta_1^2 -\beta_2^2) \leq \frac{1}{2} ( (\beta_2 - \beta_1)^2 - \beta_1^2 -\beta_2^2) = -\beta_1 \beta_2 < - \beta_1^2.
\end{equation}
Therefore, $-\beta_2^2 < \mu^* < -\beta_1^2.$ Since the quadratic $Q(\mu)$ has nonnegative discriminant, is positive at the endpoints of $[-\beta_2^2,-\beta_1^2],$ and its minimum occurs in the interior of the interval, both real roots of $Q(\mu)=0$ must also lie inside the interval $[\lambda_2, \lambda_1].$ Moreover, the bound on $u$ is
\begin{equation}
|u| \leq \beta_2 - \beta_1.
\end{equation}
Therefore, the genus-one  solution of the defocusing mKdV equation satisfies the hypotheses of Theorem~\ref{globaldefocusing} and the sharp amplitude bound of Theorem~\ref{sharpdefocusing}.

Solving the dynamical equation for $u,$ 
\begin{equation}
u_x =  \mu_2 - \mu_1, 
\end{equation} 
in the degenerate case, where $\beta_1=0$ and $\beta_2=\beta,$ gives the traveling-wave profile of the kink solution of the defocusing mKdV equation~(\ref{mkdv}), 
\begin{equation}
u(x) = \beta \tanh \left(  \beta (x-x_0) \right).
\end{equation}
In the degenerate limit, the one-phase bounded elliptic solution degenerates into the classical kink profile, which is a heteroclinic orbit connecting the degenerate equilibria. 
In the  limit, the turning points of the elliptic solution are pushed to the asymptotic equilibria $u= \pm \beta,$ which are approached by the heteroclinic orbit as $x \rightarrow \pm \infty.$ Consequently, the kink continues to obey the same amplitude bound predicted by Theorem~\ref{sharpdefocusing}. But it is the degenerate equilibrium solutions at $u=\pm \beta$ connected by the kink that attain the sharp amplitude bound. 

\section{Conclusion}

In this paper a direct finite-dimensional proof of a sharp upper bound for the amplitudes of finite-gap solutions of the modified Korteweg–de Vries equation is established. The argument is based entirely on commuting polynomial flows, invariant spectral polynomials, and elementary algebraic properties of the associated dynamical variables. The critical-point analysis leads to a natural factorization of the invariant polynomial in the maximizing configuration, allowing the sharp upper bound to be obtained without the use of theta functions. The resulting amplitude formula is expressed solely in terms of the upper-half-plane square roots of the roots of the invariant polynomial and is shown to be sharp through the explicit construction of a maximizing configuration. The proof of the sharp amplitude formula only requires the existence of a  maximizing solution and avoids the need to show that every bounded finite-gap trajectory attains its maximum. This provides a streamlined finite-dimensional framework for amplitude optimization in integrable systems associated with $\mathfrak{sl} (2, \mathbb{C})$ loop algebras.

\end{document}